\documentclass[english,nolineno]{socg-lipics-v2021}
\usepackage{graphicx} 
\usepackage{cite}
\usepackage{algorithm}
\usepackage[noend]{algpseudocode}
\newcommand{\xxx}{\hspace*{1.5em}}
\usepackage{
  amsmath, amsthm, amssymb, mathtools, dsfont, units,       
  graphicx, wrapfig, float,                                 
  listings, color, inconsolata, pythonhighlight,            
  hyperref, enumerate, framed,  
  braket}
  \usepackage{ifthen}
\newboolean{PrintComment}
\setboolean{PrintComment}{true}

\renewcommand{\emph}[1]{\textbf{\textit{#1}}}

\newcommand{\blue}[1]{{\color{blue}#1}}
\newcommand{\sandy}[1]{\ifthenelse{\boolean{PrintComment}}{\blue{Sandy: #1}}{}}
\newcommand{\shion}[1]{\ifthenelse{\boolean{PrintComment}}{\blue{Shion: #1}}{}}
\newcommand{\mike}[1]{\ifthenelse{\boolean{PrintComment}}{\blue{Mike: #1}}{}}

\title{Quantum Combine and Conquer and Its
Applications to Sublinear Quantum Convex Hull and Maxima Set Construction}
\titlerunning{Quantum Combine and Conquer}

\author{Shion Fukuzawa}{University of California, Irvine, USA \and \url{https://www.shionfukuzawa.com/}}{fukuzaws@uci.edu}{}{}

\author{Michael T. Goodrich}{University of California, Irvine, USA \and \url{https://ics.uci.edu/~goodrich/}}{goodrich@uci.edu}{https://orcid.org/0000-0002-8943-191X}{}

\author{Sandy Irani}{University of California, Irvine, USA \and \url{https://ics.uci.edu/~irani}}{irani@ics.uci.edu}{}{}

\authorrunning{Fukuzawa, Goodrich, and Irani}

\Copyright{Shion Fukuzawa, Michael T. Goodrich, and Sandy Irani} 

\ccsdesc[500]{Theory of computation}

\keywords{quantum computing, computational geometry, divide and conquer, convex hulls, maxima sets}

\funding{This work was supported in part by NSF Grant 2212129.}

\nolinenumbers 

\EventEditors{Oswin Aichholzer and Haitao Wang}
\EventNoEds{2}
\EventLongTitle{41st International Symposium on Computational Geometry
(SoCG 2025)}
\EventShortTitle{SoCG 2025}
\EventAcronym{SoCG}
\EventYear{2025}
\EventDate{June 23--27, 2025}
\EventLocation{Kanazawa, Japan}
\EventLogo{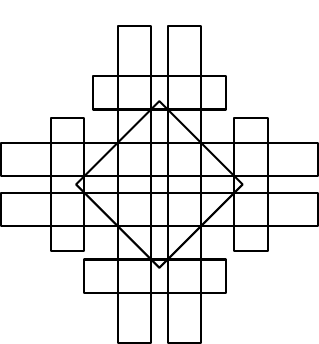}
\SeriesVolume{332}
\ArticleNo{XX}     

\begin{document}

\maketitle

\setcounter{page}{0}
\begin{abstract}
We introduce a quantum algorithm design paradigm 
called \emph{combine and conquer}, 
which is a quantum version of the ``marriage-before-conquest''
technique of Kirkpatrick and Seidel.
In a quantum combine-and-conquer algorithm,
one performs the essential
computation of the combine step of a quantum divide-and-conquer 
algorithm prior to the conquer step while avoiding recursion.
This model is better suited for the quantum setting, 
due to its non-recursive nature. 
We show the utility of this approach by providing quantum algorithms
for 2D maxima set and convex hull problems
for sorted point sets running in $\tilde{O}(\sqrt{nh})$ time, w.h.p.,
where $h$ is the size of the output. 
\end{abstract}

\section{Introduction}

In classical computing, \emph{divide and conquer} is an algorithm design
pardigm that is usually described in terms of the following three steps:

\begin{enumerate}
\item
\emph{Divide}: divide a problem into two
or more subproblems.
\item
\emph{Conquer}: solve the subproblems, typically using recursion.
\item
\emph{Combine}: 
combine subproblem solutions into a solution 
to the original problem. 
\end{enumerate}

Perhaps because of the wide applicability of algorithm design paradigms,
like divide and conquer,
there is interest in adapting classical  paradigms to the quantum setting.
For example, Childs, Kothari, Kovacs-Deak, Sundaram, and Wang~\cite{childs}
describe an adaptation of the divide-and-conquer technique
to quantum computing that in some cases results in query complexities
better than what is possible classically.
Likewise, Allcock, Bao, Belovs, Lee, and Santha~\cite{allcock}
study the time complexity of a number of quantum 
divide-and-conquer algorithms, establishing conditions under which 
search and minimization problems with classical divide-and-conquer 
algorithms are amenable to quantum speedups.
Akmal and Jin~\cite{akmal}
adapt classical divide-and-conquer approaches for string problems 
to the quantum setting.
There is also related work by many others; see, e.g.,
\cite{araujo2021divide,li,wang2022note,Chen2024,cameron2024,gong2024quantum,tomesh}.

A well-known problem that can be solved classically
using the divide-and-conquer paradigm is the
\emph{convex hull} problem;
e.g., see Seidel~\cite{seidel2017convex}.
In the two-dimensional version of this 
problem, one is given a set, $S$, of $n$ points in the plane and asked to
output a representation of the smallest convex polygon that contains
the points in $S$.
(See Figure~\ref{fig:convex}.)

\begin{figure}[hbt]
\centering
\includegraphics[width=.5\textwidth]{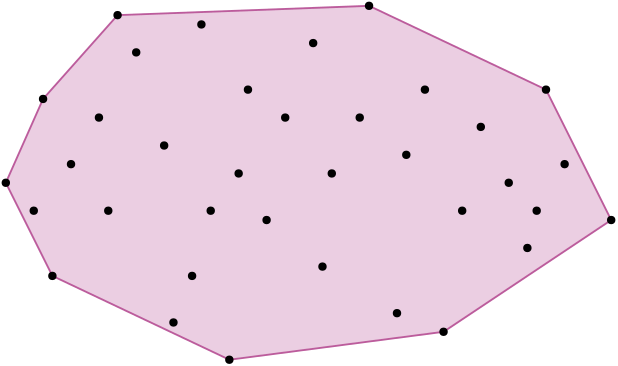}
\caption{\label{fig:convex} A two-dimensional convex hull. }
\end{figure}


Asymptotically fast classical convex hull algorithms
are \emph{output sensitive}, meaning that their running times depend
on both the input size, $n$, and output size, $h$, 
and there are well-known examples that run in
$O(n\log h)$ time; e.g., see Kirkpatrick and Seidel~\cite{kirkpatrick},
Chan~\cite{chan}, and Afshani, Barbay, and Chan~\cite{afshani2017instance}.
In this paper, we are interested in studying how to
efficiently adapt classical output-sensitive divide-and-conquer algorithms,
e.g., for convex hull and related problems, to quantum
computing models.

\bigskip
\textbf{Additional Related Work.} \hspace*{1em}
There is considerable interest in quantum algorithms
for solving computational geometry problems.
Sadakane, Sugawara, and Tokuyama~\cite{sadakane2002quantum}
give a quantum convex hull algorithm that runs in $\tilde{O}(h\sqrt{n})$
time,\footnote{We use $\tilde{O}(f(n))$ to denote $O(f(n)\log^c n)$,
   for some constant $c\ge 0$.}
and this time bound is also achieved by 
Wang and Zhou~\cite{Wang_2021}.
Note that if $h=n$, as can occur, for instance, with points distributed
on a circle or parabola, then these algorithms run in 
$\tilde{O}(n^{3/2})$ time, which is significantly slower than the best
classical algorithms.
Furthermore, the above quantum algorithms make use of a quantum input data model that assumes that the points are provided in superposition. 
Many of the procedures require sampling as well, meaning that the algorithms require multiple copies of these input states. 
We present the algorithm in a standard quantum query model that allows for an even comparison with the existing literature on the subject. In this model, we assume the existence of an oracle that returns the item in the queried index.
For completeness, Cortese and Braje \cite{cortese2018} showed that there exists a quantum circuit which can take classical data and implement this oracle, with the cost of a logarithmic overhead if the circuit construction time is excluded.
Lanzagorta and Uhlmann~\cite{lanzagorta2004quantum,lanzagorta2010quantum} 
also describe
a quantum convex hull algorithm running in $\tilde{O}(h\sqrt{n})$ time, 
which, in their case, is a quantum implementation of the well-known 
Jarvis march
convex hull algorithm using Grover search for the 
inner loop; see, e.g., \cite{deberg,seidel,preparata,orourke,grover1996fast}.
They also describe an algorithm that runs in $\tilde{O}(\sqrt{nh})$ time, 
but this result only applies when $h$ is small
and the points are chosen uniformly at random from a bounded 
convex region.
In addition to this quantum prior work, 
there is a well-developed body of 
prior work in classical computing models
for computing convex hulls where the points are given
pre-sorted in lexicographic order; see, e.g.,
Ghouse and Goodrich~\cite{ghouse1991place},
Hayashi, Nakano, and Olarlu~\cite{hayashi},
Berkman, Schieber, and Vishkin~\cite{berkman},
Nakano~\cite{nakano}, 
Goodrich~\cite{GOODRICH1993267},
and Nakagawa, Man, Ito, and Nakano~\cite{nakagawa}. 
The question we seek to answer in this paper is whether it is possible to leverage this assumption in the quantum setting and obtain asymptotically faster output-sensitive algorithms
in the model where the input point set is pre-sorted but with no other assumptions about the distribution of points.

\bigskip
\textbf{Our Results.} \hspace*{1em}
In this paper, we introduce \emph{quantum combine and conquer}, which
is a novel quantum divide-and-conquer paradigm where we perform the 
combine step \emph{before} the conquer step and avoid recursion.
Our combine-and-conquer paradigm provides
a quantum analogue to a classical divide-and-conquer 
variation that Kirkpatrick and Seidel 
call ``marriage-before-conquest''~\cite{kirkpatrick3,kirkpatrick2,kirkpatrick}.
We show that using this paradigm it is possible to compute the convex hull of a set of points in the plane in $\tilde{O}(\sqrt{nh})$ time given that the input is presorted in lexicographic order. 
As the sorting problem can be reduced to computing the convex hull, 
and there is no quantum speed-up for sorting \cite{Hoyer2001},
we don't expect quantum algorithms to outperform classical algorithms for this problem in the general case, but this result shows that some reasonable conditions can provide significant speedups. 
For example, for inputs where $h$ is $O(\log n)$, such as is expected
for $n$ points uniformly distributed in a convex polygon with
a constant number of 
sides~\cite{Har_Peled_2011_expected_size_of_random_convex_hull},
then our method achieves an $\tilde{O}(n^{1/2})$ running time.
Similarly, 
for inputs where $h$ is $O(n^{1/3})$, such as is expected
for $n$ points uniformly distributed in a disk~\cite{Har_Peled_2011_expected_size_of_random_convex_hull, Raynaud_1970_expected_size_convex_hull_disk}, 
then our method achieves an $\tilde{O}(n^{2/3})$ running time.
In general, our algorithms achieve sublinear running times
for sorted input sets for $h$ up to almost $n$.
Moreover, our results do not depend on any assumptions about the input
points coming from a given distribution, such as uniformly distributed
points in a convex region. Indeed, our results hold for adversarial chosen
point sets, such as $n$ points on a circle.

To introduce the combine-and-conquer paradigm, we first apply it to 
constructing the maxima set of a set of 
lexicographically
sorted points in the plane.
Our quantum maxima set algorithm runs in $\tilde{O}(\sqrt{nh})$ time, 
where $h$ is the size of the output.
The key idea to our technique is that there is information that can first be computed globally (the combine step), which partitions the computation into the smaller blocks 
that then use this information to finish the computation locally and 
(in a conquer step) without recursion. 

We then apply our quantum combine-and-conquer paradigm to
the convex hull problem. In this case, the combine step is more complex,
in that it includes multiple reductions to two-dimensional linear programming. 
Nonetheless, we
show that the combine steps can still be performed before the conquer
steps in this case.
As a result, we give
a quantum convex hull combine-and-conquer algorithm that runs
in $\tilde{O}(\sqrt{nh})$
time, where $n$ is the number of (sorted) input points and $h$ is 
the size of the output.
The quantum combine-and-conquer
technique provides a novel algorithm design paradigm that may be useful for designing more quantum algorithms that use similar intuition as classical algorithms.

We compare our bounds to previous quantum convex hull algorithms
in Table~\ref{tab:survey}.
In 
our algorithm's input model, the data is encoded by a unitary, such that
if the input is a list of points, $[p_0, p_1, \ldots, p_{n-1}]$, then we assume access to a unitary $U_{in}$ that maps an index $i$ to the corresponding element in the array $p_i$: $\ket{i}\ket{0} \overset{U_{in}}{\rightarrow} \ket{i}\ket{p_i}$.
The existing literature on quantum solutions to this problem assume that the data is given in a uniform superposition state, $(1/\sqrt{n}) \sum_{i} \ket{i}\ket{p_i}$. 
Access to the unitary $U_{in}$ is a somewhat stronger assumption in 
that the uniform superposition state can be prepared in unit time 
given access to~$U_{in}$.
On the other hand, the most natural way to prepare such a state 
is via access to~$U_{in}$, as we do.

\begin{table}[htb]
\begin{center}
\begin{tabular}{|c|c|c|}
\hline
                             & \textbf{Input point set assumptions}                  & \textbf{Runtime} 
\\ \hline
%
%
Sadakane, Sugawara, and Tokuyama~\cite{sadakane2002quantum}
& Unsorted, arbitrary point set
& 
\rule[-6pt]{0pt}{17pt} 
$\tilde{O}(h\sqrt{n})$ \\ \hline
Wang and Zhou~\cite{Wang_2021}
& Unsorted, arbitrary point set
& 
\rule[-6pt]{0pt}{17pt} 
$\tilde{O}(h\sqrt{n})$ \\ \hline
Lanzagorta and Uhlmann~\cite{lanzagorta2004quantum,lanzagorta2010quantum} 
& Small $h$, uniformly distributed points       
&    
\rule[-6pt]{0pt}{17pt} 
$\tilde{O}(\sqrt{nh})$                    \\ \hline
\textbf{This work}                    & Sorted, arbitrary point set 
& 
\rule[-6pt]{0pt}{17pt} 
$\tilde{O}(\sqrt{nh})$ \\ \hline
\end{tabular}
\end{center}
\caption{A summary of quantum algorithms to compute the convex hull of a set of $n$ points, where the output has $h$ edges. 
The algorithms all assume access to a data loading unitary.}
\label{tab:survey}
\end{table}

\section{Preliminaries}


We assume basic familiarity with quantum computing;
see, e.g., Nielsen and Chuang~\cite{nielsen2010quantum}.
A \textbf{quantum state} over $n$ qubits is a unit vector of length $2^n$ with complex entries. The computational basis $\{\ket{j}\}_{j \in [0,\ldots,2^n-1]}$ is a basis over quantum states where $\ket{j}$ represents the unit vector of length $n$ with a 1 in the $j$-th index and 0 elsewhere. A computational basis state can be used as a classical bit string to simulate any classical algorithm. We can express any quantum state as a weighted sum of these classical basis states,
\begin{equation}
    \ket{\Psi} = \sum_{j = 0}^{2^n - 1} \alpha_j\ket{j}. 
\end{equation}
If at least two $\alpha_j$ in the above are nonzero, we say the state is in \textbf{superposition}. 
A \textbf{measurement} of the state $\ket{\Psi}$ will collapse the state to $\ket{j}$ with probability $|\alpha_j|^2$. The measurement is destructive in the sense that once it is collapsed, there is no way to go back to $\ket{\Psi}$ without running the circuit to prepare that state. The state where all $\alpha_j = \frac{1}{\sqrt{2^n}}$ can be prepared by a circuit of depth 1. 


We will assume an input model where the input data $[p_0, p_1, \ldots, p_{n-1}]$ is accessible by a unitary $U_{in}$ that maps an index $i$ to the corresponding element in the array $p_i$. Since quantum operations must be reversible, it is standard to model the action of this unitary by using an index and data register as $\ket{i}\ket{0} \overset{U_{in}}{\rightarrow} \ket{i}\ket{p_i}$.
We can use linearity to examine the action of this unitary to a state in superposition 
\begin{equation}\label{eq:qram-state}
    U_{in}\left(\sum_{j = 0}^{2^n - 1} \frac{1}{\sqrt{2^n}} \ket{j}\ket{0}\right) = \sum_{j = 0}^{2^n - 1} \frac{1}{\sqrt{2^n}}(U_{in}\ket{j}\ket{0}) = \frac{1}{\sqrt{2^n}}\sum_{j = 0}^{2^n - 1} \ket{j}\ket{p_j},
\end{equation}
from which the data point $p_j$ can be retrieved with probability $|\alpha_j|^2$ upon measurement of the index register. 
Therefore, in this model, we assume that a quantum state representing a distribution of our inputs is preparable in time proportional to the time it takes to prepare the distribution. 
A transformation from the $n$-bit zero string $\ket{0^n}$ to a uniform superposition over all $n$-bit strings can be accomplished by a depth 1 circuit where a Hadamard gate is applied to each qubit in parallel. 

Much of the existing literature on this subject assumes access to multiple copies of the state in (\ref{eq:qram-state}). Our assumption of access to $U_{in}$ is at least as strong of an assumption as having multiple copies of the state. At the same time, it is a standard assumption for states in the form of (\ref{eq:qram-state}) are generated using access to a unitary like $U_{in}$. 
We therefore compare the performance of our algorithm against the existing literature under equal access to $U_{in}$, and we summarize the performance of each in Table \ref{tab:survey} under this model. 


\textbf{Quantum Subroutines.}\hspace*{1em}
In this section we discuss the two primary quantum subroutines that will be used throughout the paper. The quantum computing components of our algorithms
are limited to invocations of these subroutines, and the remainder of the algorithms are done through classical post-processing. 

The first subroutine is for preparing superposition states of subsets of points. Our algorithms use the fact that the data $S = [p_0, p_1, \ldots, p_{n-1}]$ is encoded as a sorted list in $U_{in}$, and we need to prepare $h$ states $\ket{S_0}, \ket{S_1}, \ldots, \ket{S_{h-1}}$ where $\ket{S_j}$ is a uniform superposition of $n/h$ consecutive elements in the array starting at index $jn/h$ and $h$ is a power of 2. Preparation of the $\ket{S_j}$ begins with a register storing an $n$-bit zero string $\ket{0^n}$, and another register large enough to store a point from the dataset, each of which we will call the index and data register respectively. The first step is to set the first $\log_2 h$ bits of the index register to $j$. Next, take the remaining $n - \log_2 h$ bits in the index register and prepare a uniform superposition state over the bitstrings of length $n - \log_2 h$. Finally, apply $U_{in}$ on this state and the data register to prepare the superposition of $n/h$ states starting from $jn/h$, 
\begin{equation}\label{eq:state-in-block-j}
    \ket{S_j} = \frac{1}{\sqrt{n/h}}\sum_{k=0}^{n/h - 1} \ket{jn/h+k}_{I_j} \ket{p_{jn/h+k}}_{D_j}.
\end{equation}
We summarize this process in Algorithm~\ref{alg:prep-j}.
Throughout this paper, we use $S_j$ to represent the subset of $S$ storing points in the $j$-th block, and $\ket{S_j}$ to be the quantum state that encodes this set, as seen in Equation~(\ref{eq:state-in-block-j}). 

\begin{algorithm}[ht]
\caption{\texttt{qPrep}($j$, $h$)}
\label{alg:prep-j}
\begin{algorithmic}[1]
\Require $j$: The index of the block to prepare; $h$: the number of blocks to partition into.
\Require Global access to the unitary $U_{in}$ that encodes the sorted data $[d_0, \ldots, d_{n-1}]$
\State Prepare two registers, initialized to $\ket{0^n}_{I_j}\ket{0}_{D_j}$. 
\State Set the first $\log_2 h$ bits of $I_j$ to $j$
\State Apply a Hadamard gate to each remaining qubits in $I_j$. 
\State Apply $U_{in}$ to the pair of registers. 
\State \Return $\ket{S_j}$, the state 
after the previous step as described in equation \ref{eq:state-in-block-j}. 
\end{algorithmic}
\end{algorithm}



The second quantum subroutine we use is 
a quantum max/min finding algorithm due to
Durr and Hoyer~\cite{durr1999}. 
This algorithm takes as input a Boolean function, 
$f : D \rightarrow \{0, 1\}$, a comparator, $\preceq$, 
to maximize (or minimize) over, the index, $j$, 
of the block to perform the search over, 
and the total number of blocks, $h$. 
Our subroutine uses \texttt{qPrep}$(j)$ to prepare the superposition state described in Theorem \ref{thm:min-finding}. This has the effect of applying the algorithm only to the points in block $j$, and allows us to search the block in $\tilde{O}(\sqrt{n/h})$ time.
Throughout the text, we use the signature,
$
    \texttt{qMax/qMin($f$, $\preceq$, $j$, $h$)},
$
when calling this function, and the output is an element of $D$ or null if there are no elements such that $f(d) = 1$. The function $f$ can be passed as, say, a classical circuit which can be converted to a quantum circuit. 

\begin{theorem}[Quantum Maximum/Minimum Finding \cite{durr1999}]\label{thm:min-finding}
    Let $D = [d_0, \ldots, d_{n-1}]$ be a list of $n$ elements represented by $w$ bits, and let $S$ be the time required to prepare the state,
    \begin{equation}
        \ket{\psi} = \frac{1}{\sqrt{n}}\sum_{i=0}^n \ket{i}\ket{d_i},
    \end{equation}
    and $Q$  the time it takes to query a boolean function $f: D \rightarrow \{0, 1\}$. Let $M$ be the subset of $D$ such that $f(m) = 1$ for all $m \in M$. Also, let $\preceq$ be some ordering of data values in the data register such that comparisons according to this ordering can be performed in $O(\log w)$ time. 
    Then we can find the maximum (or minimum) value of $M$ under the specified ordering in time $S\cdot Q \cdot \Tilde{O}( \sqrt{n})$.
\end{theorem}

%




\section{Quantum Maxima Sets}

In this section, we present a quantum algorithm to solve the maxima set problem using the combine-and-conquer paradigm. 
In the maxima set problem, we are given a set, $S$, of $n$ points in the plane and are tasked with finding the subset of points in $S$ that are not dominated by any other points in $S$. We say that a point $p$ \textbf{dominates} a point $q$ if $p.x > q.x$ and $p.y > q.y$.
Note that the output size, $h$, can be as small as $1$ or as large as $n$.

We assume that the input set is sorted in lexicographic order (first by
$x$-coordinates and then by $y$-coodinates in the case of ties), as discussed in the introduction. 

Let us begin by reviewing a simple divide-and-conquer algorithm in 
the classical computing model,
which is  provided the set, $S$, of $n$ points as input in sorted order,
and returns the maxima set, $M$  (see Preparata and
Shamos~\cite{preparata}).
The algorithm splits $S$ into left and right sets, $S_1$
and $S_2$, and recursively finds the maxima sets of each.
Then it prunes away each point in $S_1$ with smaller $y$-coordinate
than the point, $p_{\rm max}$, in $S_2$ with maximum $y$-coordinate.
The resulting maxima set algorithm,
which we give in detail in an appendix, 
takes $\Theta(n\log n)$ time, even when $S$ is sorted~\cite{preparata}.
Kirkpatrick and Seidel~\cite{kirkpatrick2} show that if the combine
step of finding $p_{\rm max}$ 
is performed before the divide and conquer steps, and used to prune
points in recursive subproblems, then the running time
can be improved to $O(n\log h)$.

For our output-sensitive quantum algorithm,
we also perform the combine step before the conquer step, but
our goal is to design a quantum 
algorithm that runs in $\tilde{O}(\sqrt{nh})$ time
without recursion,
which requires a more sophisticated approach. We first describe an algorithm where the output size, $h$, is known in advance and later show how to use this algorithm to solve the case where the output size is not known. 

At a high level, the combine-and-conquer approach performs the combine step first before conquering the subproblems. In the maxima set problem, this is achieved in the following way. 
We begin by dividing our set $S$ into $h$ different blocks using Algorithm \ref{alg:prep-j}. 
The key observation we use is that we can isolate the problem of finding the maxima set to each of the local blocks once we identify representative information from each block. This representative information is used globally, so we describe the process of finding such information as the \textbf{combine step}. The combine-and-conquer paradigm works for problems where, once the global information is identified, the blocks do not need to communicate with each other for the remainder of the algorithm. In the case of computing the maxima set, the combine step consists of searching for the set $T$ of the $h$ highest points in each block. Once the set $T$ is found, we classically compute the maxima set $M_T$ of $T$. Note that $M_T$ cannot contain more than $h$ points, and if a block does not share a point with $M_T$ no points in that block can be in the final output set. 
Once the points in $M_T$ are identified, we begin the \textbf{conquer step} of the algorithm to find the points in each block that appear in the maxima set. Block number $j$ uses information from two points identified in the combine step to restrict the search region. The first is the $y$ coordinate $T_j$ of the tallest point within the $j$-th block, which immediately removes all points to the left of it inside block $j$ out of the candidate set. The second is the $y$ coordinate $R_j$ of the highest point that appears in any of the blocks to the right of block $j$.
Note that $R_j$ is left-most point from $M_T$ that is to the right of block $j$.
See Figure~\ref{fig:maxima-set}.

\begin{figure}[hbt!]
    \centering
    \includegraphics[width=.9\linewidth]{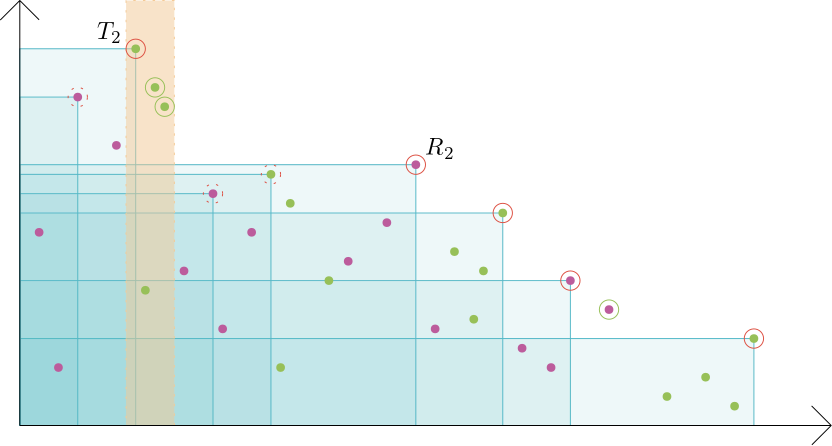}
    \caption{An example instance of maxima set where $n = 32$ and $h = 8$ solved using our quantum combine-and-conquer algorithm. Consecutive blocks of points are colored in alternating purple and green, and the set $T$ of tallest points in each group is circled. The maxima set of $T$ is indicated using solid circles, whereas tallest points that are dominated are circled in dotted lines. To illustrate an example of the conquer step, we focus attention on the group $S_2$ which is highlighted in orange. In this group, the tallest $y$ coordinate is indicated by $T_2$, and $R_2$ is the $y$ coordinate of the tallest point to the right of this group. Thus, the only points that need to be processed are the two points not bound by any blue box. The points that are found in this local check are circled in green. }
    \label{fig:maxima-set}
\end{figure}

Using $T_j$ and $R_j$, we can define a Boolean function $f$  that takes as input a point $p$ and returns 1 if $p$ is not dominated by either of $T_j$ and $R_j$. We then apply Theorem \ref{thm:min-finding} to search for the maximum  element in lexicographic order subject to this Boolean function. It turns out that this point is part of the final output set, as $R_j$ tells us it is not dominated by any point in blocks to the right of $j$ and $T_j$ tells us it is not dominated by any point in block $j$. We then update $R_j$ to equal this maximal point we found, and repeat this process in block $j$ until the Boolean function does not mark any elements in the set. We illustrate the first iteration of this process in Figure~\ref{fig:maxima-set-block}. Since the output size is $h$, it suffices to run this process $O(h)$ times in total across all blocks. Finally, we collect the outputs in the blocks that contain a point in $M_T$, which takes at most $O(h)$ time.  
See Algorithm~\ref{alg:complete-maxima-set}. 


\begin{figure}[hbt!]
     \centering
     \begin{subfigure}[b]{0.45\textwidth}
         \centering
         \includegraphics[width=\textwidth]{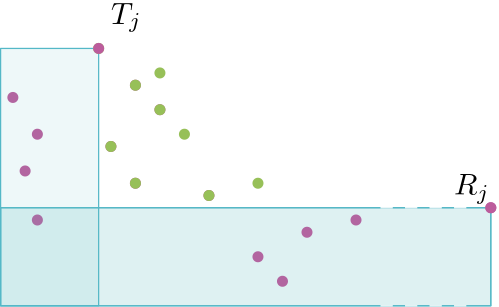}
         \label{fig:maxima-set-block-1}
     \end{subfigure}
     \hfill
     \begin{subfigure}[b]{0.45\textwidth}
         \centering
         \includegraphics[width=\textwidth]{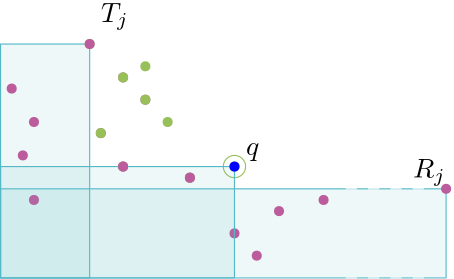}
         \label{fig:maxima-set-block-2}
     \end{subfigure}
        \caption{A demonstration of Algorithm \ref{alg:complete-maxima-set} where we iteratively search for the maxima point within block $j$ given $T_j$ and $R_j$. The left shows the state of the first iteration, and all points shown are the contents of block $j$ except for the one labeled $R_j$.  Points are green if they are not dominated by $T_j$ or $R_j$ and thus return 1 to the Boolean function $f$. Among these points, we search for the lexicographic maximum point $q$ and set this to be the new right boundary for the next iteration. The first $q$ to be discovered in the above instance is marked in blue. This is repeated until there are no remaining green points. }
        \label{fig:maxima-set-block}
\end{figure}

\begin{algorithm}[htb!]
\caption{\texttt{CompleteMaximaSet}($j$, $h$, $T_j$, $M_j$)}\label{alg:complete-maxima-set}
\begin{algorithmic}[1]
\Require $j, h$: The block we are conquering and the total number of blocks. 
\Require $T_j$: The $y$ coordinate of the tallest point in $S_j$. 
\Require $R_j$: The $y$ coordinate of the first point in $M$ appearing after all points in $S_j$. 
\State Define a Boolean function $f$ that takes as input a point $p$ and returns 1 if $p$ is not dominated by $T_j$ or $R_j$. 
\State {\bf while} \texttt{qMax($f$,$\preceq_{Lex}$, $j$, $h$}) returns a point $q$ {\bf do}
\State \xxx Store $q$ in an output list.  
\State \xxx Define a new $f$ that takes as input a point $p$ and returns 1 if $p$ is not dominated by\break 
\xxx $T_j$ or $q$. 
\end{algorithmic}
\end{algorithm}

\begin{algorithm}[ht]
\caption{\texttt{QuantumMaximaSet}($S$)}\label{alg:maxima-set-info}
\begin{algorithmic}[1]
\Require $U_{in}$: a the unitary that encodes the lexicographically sorted data $[d_0, \ldots, d_{n-1}]$
\Require $h$: the number of blocks to partition into.
\State $h = 4$.
\State {\bf while} True {\bf do}
\State \xxx {\bf for} $j \in [0, 1, \ldots, h-1]$
\State \xxx \xxx $T_j = \texttt{qMax}(f, \preceq_y, j, h)$ where $f(x) = 1$ for all $x \in S$ 
\State \xxx Let $T = [T_0, T_1,\ldots, T_{h-1}]$.
\State \xxx $M = \texttt{ClassicalMaximaSet($T$)}$ 
\State \xxx {\bf for} $j \in [0, 1, \ldots, h-1]$
\State \xxx \xxx \texttt{CompleteMaximaSet($j$, $h$, $T_j$, $M_j$)}
\State \xxx {\bf if} {The total number of points found exceeds $h$} {\bf then}
\State \xxx\xxx $h = 2h$
\State \xxx {\bf else}
\State \xxx\xxx Print all the points in the output register of all $j$ blocks and terminate the loop.
\end{algorithmic}
\end{algorithm}

Finally, the case where $h$ is unknown can be handled using an approach where we start with a small estimate of $h$, then repeat the procedure using $h = 2h$ if we discover more than $h$ points. 
Thus, the running time forms a geometric sum that is 
$\tilde{O}(\sqrt{nh})$.
See Algorithm \ref{alg:maxima-set-info}. 

\begin{theorem}
There is a quantum algorithm for finding
    the maxima set of a presorted set of $n$ points in the plane 
in $\tilde{O}(\sqrt{nh})$ time, where $h$ is the size of the output.
\end{theorem}

\begin{proof} 
The Grover searches for the tallest points in each block takes $\tilde{O}(h\sqrt{n/h} ) = \tilde{O}(\sqrt{nh})$ time. 
    Once this set $T$ is found, the maxima set $M$  of $T$ can be computed classically in $O(h)$ time. The total time incurred by  the calls to
    \texttt{CompleteMaximaSet} 
    is dominated by the calls to \texttt{qMax}, each of which outputs a point or terminates the call to  \texttt{CompleteMaximaSet}.
    Therefore \texttt{qMax} is called $O(h)$ times and each call requires $\tilde{O}(\sqrt{n/h})$ time for a total of $\tilde{O}(\sqrt{nh})$ time.
 Note that the condition that the points we are searching over are between $T_j$ and $R_j$ can be done in constant time. 
The outer loop repeats at most $O(\log n)$ times, whose contribution is suppressed by $\tilde{O}$. 
    Finally, the time to output the points that were found takes $O(h)$ time and is dominated by the other parts of the algorithm.
\end{proof}

\section{Quantum Convex Hulls}

In this section, we describe a quantum algorithm to construct the convex hull 
of $n$ points in the plane
in $\tilde{O}(\sqrt{nh})$ time, where $h$ is the size of the output. 
Again, we assume that the input set is sorted in lexicographic order. 
We follow our combine-and-conquer paradigm, where we use some global 
properties that can be computed from our subsets to prune off points 
that do not appear in the output, 
then use these values to isolate the problem solving within each block. 
Here, however, the combine step is more complicated that simply finding
a point, $p_{\rm max}$, with maximum $y$-coordinate.

We begin by reviewing a well-known classical divide-and-conquer 
algorithm for 2D
convex hulls (see, e.g., \cite{deberg,seidel,preparata,orourke}).
The convex hull $CH(S)$ of a set, $S$, of $n$ points in the plane
can be described as a union of the polygonal chains
that form the upper hull $UH(S)$ and lower hull $LH(S)$ of the points, where
$UH(S)$ (resp., $LH(S)$) is the set of edges of $CH(S)$ with positive 
(negative) normals. 
For completeness, we provide a description of a classical divide-and-conquer
algorithm for computing $UH(S)$ in appendix D,
which divides $S$ into a left set, $S_1$, and right set, $S_2$,
and recursively finds $UH(S_1)$ and $UH(S_2)$. Then it finds the 
tanget segment bridge between 
$UH(S_1)$ and $UH(S_2)$ and prunes away the points under the bridge,
resulting in an algorithm running in $O(n\log n)$ time.
By a well-known point-line duality,
which we also review in an appendix,
the bridge edge 
can be found by a solving a two-dimensional linear program;
see, e.g., \cite{deberg,seidel,orourke}.

\begin{lemma}[see, e.g.,~\cite{deberg,seidel,orourke}]
\label{lemma:bridge-edge}
    Given a vertical line $L$, and a set $S$ of $n$ points in a plane, one can determine the edge $e$ of $UH(S)$ that intersects $L$ (or that no such edge exists) by finding the highest point of intersection between $O(n)$ halfplanes. 
\end{lemma}

If we use linear-time 2D linear programming for 
the bridge-finding step, then the
running time for this classical divide-and-conquer algorithm is $O(n\log n)$,
even if the points are sorted by $x$-coordinate.
Kirkpatrick and Seidel~\cite{kirkpatrick} show that this classical running
time can be improved to $O(n\log h)$ by performing the bridge-finding
step before the (recursive) conquer steps.
Likewise, our quantum convex hull algorithm also performs this
combine step before the conquer steps, but avoids recursion. 

Before we give our algorithm, let us review
another well-known classical algorithm for computing a two-dimensional
convex hull---the \emph{Jarvis march} algorithm;
see, e.g., \cite{deberg,seidel,preparata}.
In this algorithm, we start with a point, $p$, known to be on the convex 
hull, such as a point with smallest $y$-coordinate. Then, we find
the next point on the convex hull in a clockwise order, by a simple
linear-time maximum-finding step. We iterate this search, winding around
the convex hull, until we return to the starting point.
Since each iteration takes $O(n)$ time, this algorithm 
runs in $O(nh)$ time.
Lanzagorta and Uhlmann \cite{lanzagorta2004quantum}
describe a quantum assisted version of the Jarvis march algorithm
that runs in $\tilde{O}(h\sqrt{n})$ time,
which we review in an appendix.


Let $S$ be a set of $n$ points in the plane sorted lexicographically.
For simplicity of expression, let us assume that the points have
distinct $x$-coordinates and we are interested in computing the
upper hull, $UH(S)$, of $S$. 
As in our maxima set algorithm, let us assume for now that we know
the value of $h$, the number of points on the upper hull of $S$ to simplify the presentation.
We divide $S$ into $h$ blocks, $S_0,S_1,\ldots,S_{h-1}$,
of size $O(n/h)$ each by vertical lines, $L_1,L_2,\ldots,L_{h-1}$,
arranged left-to-right.
This is achieved by using Algorithm \ref{alg:prep-j} to partition the input.

At a high level, our quantum combine-and-conquer upper hull algorithm
first computes all the bridge edges that are intersected by one of the
vertical lines, $L_i$. 
This amounts to the combine step of our quantum 
combine-and-conquer algorithm.
Each of these bridge edges belong to the upper hull of $S$; hence,
there are at most $h$ such bridges. Importantly, we show how to compute
all these bridges in $\tilde{O}(\sqrt{nh})$ time.
Each block of points, $S_j$, has $O(n/h)$ points,
and, of course, there may be remaining upper hull points to compute
for each such $S_j$ (between points of tangency for the bridges computed
in the combine step).
Note that the conquer step for block $S_j$ will return
in $O(1)$ time if $S_j$ does not contain any endpoints of a bridge edge or if it contains a single point where two bridge edges meet. 
See Figure~\ref{fig:convex-hull-alg} showing the edges that come from the combine step and the edges that come from the conquer step.

To compute the remaining convex hull points, then, we perform a modified
quantum Jarvis march for each subset, $S_j$, which runs
in time $\tilde{O}(h_j\sqrt{n/h})$ time, where $h_j$ is the number
of additional upper hull points found for $S_j$.
Thus, the total time for this second (conquer) phase of our algorithm
is
\begin{eqnarray*}
\sum_{j=1}^h \tilde{O}(h_j\sqrt{n/h}) \ \ &\le& \tilde{O}(h\sqrt{n/h}) \\
&=& \tilde{O}(\sqrt{nh}).
\end{eqnarray*}

Given this overview, let us next describe how to perform the different
steps of our algorithm,
%
beginning with our quantum algorithm for finding
all the bridge edges intersecting the vertical lines, 
$L_1,L_2,\ldots,L_{h-1}$, which separate the subsets, $S_0,S_1,\ldots,S_{h-1}$,
of $S$.
At a high level,
our combine-step method can be viewed as a quantum adaptation of a classical
algorithm by Goodrich~\cite{GOODRICH1993267} for computing the upper
hull of a partially sorted set of points.

As mentioned above,
by point-line duality, bridge finding 
is dual to two-dimensional linear programming.
In more detail,
let $S$ be a set of $n$ points in the plane, and $L$ be a vertical line in this plane. Let $S'$ be the same set of points after shifting the plane horizontally so that $L$ is the $y$-axis. 
We define a dual plane by taking each point $p = (a, b)$ and mapping it to a line $y = ax - b$. 
A standard result in computational geometry is that the points in the upper hull of $S'$ correspond to the lower envelope of 
the set of lines in the dual plane; 
see, e.g.,~\cite{deberg,seidel,preparata}.
Furthermore, by the duality transformation, the highest 
point in the lower envelope is the intersection of two lines in the envelope whose slopes transition from positive to negative. Mapping this point to the primal plane gives us the bridge edge we are looking for. Thus, we can reduce the problem of finding a bridge edge to the problem of finding the largest point in the intersection of $n$ lower-halfplanes. 
We also use the following lemma due to Sadakane et al. showing that the problem of computing the highest point in the intersection of $n$ halfplanes efficiently using a quantum algorithm. 

\begin{lemma}[Sadakane et al. \cite{sadakane2002quantum}]\label{lemma:linear-program}
    The highest point in the intersection of $n$ lower-halfplanes can be computed in $O(\sqrt{n}\log^2 n)$ time, using a quantum computer, 
with probability $1 - n^{-c}$ for any constant $c$ given the linear constraints in superposition. 
\end{lemma}

The linear constraints in superposition refers to a state in the following form:
\begin{equation}
    \ket{\psi} := \frac{1}{\sqrt{n}} \sum_{i = 0}^{n - 1} \ket{i}\ket{p_i}
\end{equation}
where $p_i$ are the coordinates of each of the $n$ points. By point-line duality, the coordinates fully describe the halfplanes in the dual plane. 

The output element is a point in the dual plane, which maps back to a line in the primal plane. To determine which points this line intersects, we can run two iterations of minimum finding to find the points in $S$ that are on the resulting line by minimizing by the distance to the line. 
%
We combine the above two lemmas to a function with the signature,
\begin{equation}
    \texttt{qLP($L$, $i$, $j$, $h$)},
\end{equation}
and it returns the endpoints $p_s, p_f \in S$ of the bridge edge as a tuple. This process will also call $\texttt{qPrep}($i$, $h$)$ as needed. 

\begin{algorithm}
\caption{Bridge($i$, $j$, $h$)}\label{alg:bridge}
\begin{algorithmic}[1]
\Require $i, j$: The indices of blocks to find the bridge edge over. 
\Require $h$: The total number of blocks that we are partitioning into. 
\State Let $L$ be a vertical line drawn between the last point of $S_i$ and the first point of $S_j$. This can be computed by querying $U_{in}$ at the final and first index of blocks $i$ and $j$ respectively.
\State $(p_s, p_f) = $\texttt{qLP}($L$, $i$, $j$, $h$)
\State Return the two endpoints $p_s \in S_i$ and $p_f \in S_j$ of the bridge edge. 
\end{algorithmic}
\end{algorithm}

We are now ready to describe the full algorithm. The first step will be to divide our input set into $h$ blocks, each containing $n/h$ points. We define a \textbf{bridge edge} to be an edge in the upper hull of $S$ whose endpoints are in two different blocks. 

To find the bridge edges, we take inspiration from the classical Graham scan algorithm for computing the convex hull. We will be using the fact that if we traverse the convex hull in the clockwise direction, each edge makes a right turn from its previous edge. Since bridge edges are edges in the upper hull, this is true for them as well. We will compute bridge edges between consecutive pairs, but if at any point a left turn is formed between two bridge edges, we know that these edges will not be a part of the upper hull. Critically, this allows us to ignore an entire block of points, and move on to find the bridge edge between the remaining blocks. 
The details of the above approach 
are outlined in Algorithm~\ref{alg:bridge-edges}.

\begin{algorithm}
\caption{FindBridgeEdges($h$)}\label{alg:bridge-edges}
\begin{algorithmic}[1]
\State Initialize an empty stack $\Sigma$.
\State Let $L_1 := $Bridge($0, 1, h$), and push this onto $\Sigma$. 
\Comment{$L_i$ will represent the bridge edge that ends in block $i$.}
\State {\bf for} $i \in [2, h-1]$ {\bf do}
\State \xxx Let $L_i :=$ Bridge($i-1, i, h$).
\State \xxx {\bf while} {$L_{i - 1}$ and $L_i$ form a left turn} {\bf do}
\State \xxx\xxx Pop the stack
\State \xxx\xxx Let $t$ be the index of the item on top of the stack. 
\State \xxx\xxx Let $L_i :=$ Bridge($t, i, h$). 
\State \xxx Push $L_i$ on the stack
\State Return $\Sigma$ as an array. 
\end{algorithmic}
\end{algorithm}

\begin{lemma}\label{lemma:bridge-edges}
    Let $S$ be a set of points divided into $h$ blocks such that $p_i \leq p_j$ for all $p_i \in S_i$ and $p_j \in S_j$ where $i \leq j$. Then, the set of bridge edges of this partition can be found in $\Tilde{O}(\sqrt{nh})$ time.
\end{lemma}

\begin{proof}
    Each block's pending contribution to the bridge edges can only be popped at most once. Thus, the total time it takes for the loop starting at line 3 of Algorithm \ref{alg:bridge-edges} will be proportional to $O(h)$. For each edge, we must compute the Bridge between sets $S_j$ and $S_i$, which will take $\Tilde{O}(\sqrt{n/h})$ time by lemma \ref{lemma:bridge-edge} and \ref{lemma:linear-program}. Thus, the total time it takes to find the bridge edges is $\Tilde{O}(\sqrt{nh})$. 
\end{proof}

Once the bridge edges are identified, what remains is to find the edges of the convex hull whose end points are in the same block. 
To do this, we use a quantum assisted Jarvis march. 
Recall that the Jarvis march is another standard convex hull algorithm 
in which we start with an edge $(p, q)$ on the convex hull, 
then do a search over the remaining points to find the point $r$ 
such that the angle formed between $(p, q)$ and $(q, r)$ is maximized. 

In each block, $S_i$, rather than performing a full Jarvis march, however, 
we start the search from the bridge edge entering $S_i$ from the left, and perform up to $h$ Grover searches for the point maximizing the angle until we find the starting point of the bridge edge leaving $S_i$. 
We summarize this step in Algorithm~\ref{alg:q-jarvis-march} and Lemma~\ref{lemma:q-jarvis-march}.
See also Figure~\ref{fig:jarvis-march}.

\begin{algorithm}
\caption{QuantumBlockJarvisMarch($j$, $h$, $B$)}\label{alg:q-jarvis-march}
\begin{algorithmic}[1]
\Require $B$ is the set of bridges of $S$ computed in Algorithm \ref{alg:bridge-edges}. 
\State Define $f_c$ to be a function such that $f_c(p) = 1$ iff there is a bridge in $B$ that starts at $p$.
\State Define $f_f$ to be a function such that $f_f(p) = 1$ iff there is a bridge in $B$ that ends at $p$.
\State $p_c = \texttt{qMax}(f_c,\preceq_{Lex}, j, h)$
\State $p_f = \texttt{qMax}(f_f,\preceq_{Lex}, j, h)$
\State If neither are present, Return. 
\State Let $p_s, p_c$ be the two end points of the bridge edge containing $p_c$.
\State {\bf while} $p_c \neq p_f$ {\bf do}
\State \xxx Define $\preceq_{deg(p_s,p_c)}$ to be the ordering induced by the angle formed between $(p_s, p_c)$ \break \xxx and $(p_c, p)$ for $p \in S$. 
\State \xxx $p_r = \texttt{qMax}(f, \preceq_{deg(p_s,p_c)}, j, h)$ where $f(x) = 1$ for all $x$.
\State \xxx Add $p_r$ to the set of output points associated with block $j$.
\State \xxx Let $(p_s, p_c) = (p_c, p_r)$. 
\end{algorithmic}
\end{algorithm}

\begin{figure}[htb]
    \centering
    \includegraphics[width=.95\linewidth]{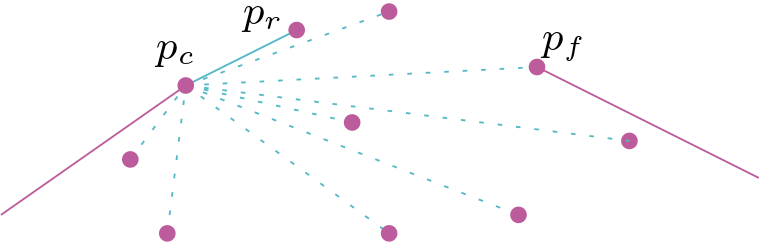}
    \caption{ An example of our restricted Jarvis March convex hull algorithm. 
     The purple lines are edges that we know are on the convex hull, and the points are the content of some block $S_j$. In a Jarvis March algorithm, we start our search from $p_c$ which is known to be on the convex hull, then search for the point that forms the maximum angle with the incoming edge containing $p_c$. In the above example, the edges we check are in blue, and the solid line denotes the edge that is added to the convex hull. We repeat the search again starting at $p_r$, until we connect to $p_f$.}
    \label{fig:jarvis-march}
\end{figure}


\begin{lemma}
\label{lemma:q-jarvis-march}
    Let $S$ be a set of $m$ points and $p_s$, $p_f \in S$ be points on the convex hull of $S$. The part of the convex hull from $p_s$ to $p_f$ in the clockwise direction can be computed in $\Tilde{O}(h\sqrt{m})$ time, where $h$ is the number of points on the convex hull from $p_s$ to $p_f$. 
\end{lemma}


We now combine the above pieces to describe our complete
quantum convex hull algorithm. 
See Figure~\ref{fig:convex-hull-alg}.
As was the case in the maxima set algorithm, we begin with a guess for the upper bound of $h$, then during the conquer step detect whether or not we have exceeded this bound. 
A convex hull requires at least 3 points, so we use 4 as our initial guess, as it is the smallest power of 2 greater than 3. 
If we exceed our guess, then we double our estimate for $h$ and rerun the algorithm. 
Thus, the running times across all our guesses form a geometric sum that is $\tilde{O}(\sqrt{nh})$.
See Algorithm~\ref{alg:convex-hull-h}.

\begin{figure}[htb]
    \centering
    \includegraphics[width=.75\linewidth]{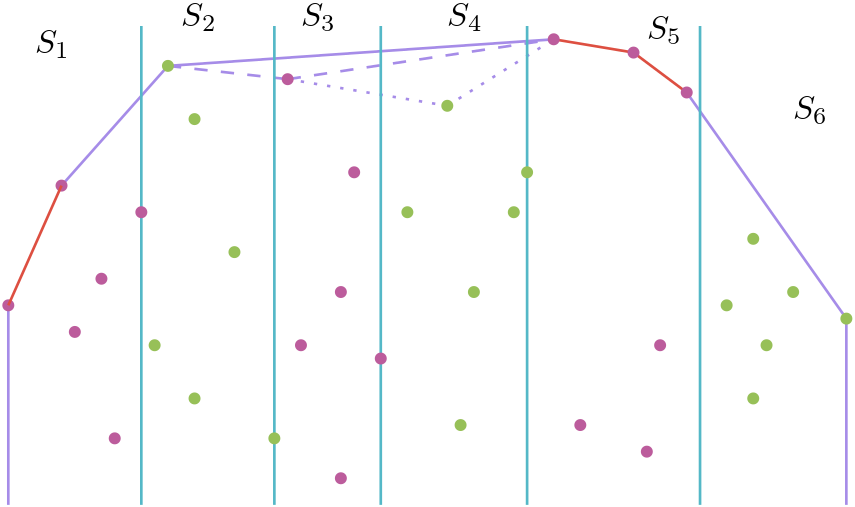}
    \caption{An upper hull with $n = 36$ and $h = 6$. Consecutive blocks of points are colored in alternating purple and green. The bridge edges discovered are shown in purple, and the dotted and dashed lines starting from block $S_2$ indicate how the algorithm handles left turns. The bridge edge between $S_4$ and $S_5$ forms a left turn relative to the previous dotted bridge edge, so both are popped and a new bridge edge is found between $S_3$ and $S_5$. This is repeated until no left turns are formed, giving the solid purple line. Finally, the bridge edges do not close the hull in some blocks, which is where we run the quantum Jarvis march to find the edges of the upper hull shown in red. }
    \label{fig:convex-hull-alg}
\end{figure}

\begin{algorithm}[hbt]
\caption{QuantumConvexHull($U_{in}$)}
\label{alg:convex-hull-h}
\begin{algorithmic}[1]
\Require $U_{in}$: A unitary that encodes the lexicographically sorted data $[d_0,\ldots,d_{n-1}]$
\State Let $h = 4$. 
\State {\bf while} True {\bf do}
\State \xxx $B$ = \texttt{FindBridgeEdges($h$)}
\State \xxx {\bf for} $j \in [0, 1, \ldots, h-1]$ {\bf do}
\State \xxx\xxx \texttt{QuantumBlockJarvisMarch}($j, h, B$)
\State \xxx {\bf if} {the total number of points found exceeds $h$} {\bf then}
\State \xxx\xxx $h = 2h$
\State \xxx {\bf else}
\State \xxx\xxx Print all the points in the output register of each block and terminate the loop.
\end{algorithmic}
\end{algorithm}

\begin{theorem}
There is a quantum algorithm
    for finding the convex hull of a set of $n$ points in lexicographic order 
that runs in $\Tilde{O}(\sqrt{nh})$ time, where $h$ is the size of
the output.
\end{theorem}

\begin{proof}
    As observed, the outer loop will run at most $\log n$ times. 
    By lemma \ref{lemma:bridge-edges}, \texttt{FindBridgeEdges} takes $\Tilde{O}(\sqrt{nh})$ time. By lemma \ref{lemma:q-jarvis-march}, \texttt{QuantumJarvisMarch} takes $\Tilde{O}(h_j \sqrt{n/h})$ time
    for block $S_j$, where $h_j$ is the number of points from $S_j$ that lie on the convex hull of the entire points set $S$. 
    Since $\sum_j h_j \le h$, the total running time for all the calls to \texttt{QuantumJarvisMarch} is $\Tilde{O}( \sqrt{nh})$. 
    Finally, printing the $h$ discovered points takes $O(h)$ time. 
\end{proof}

In an appendix, we give a lower bound for quantum convex hull that shows
our algorithm is optimal up to polylogarithmic factors.

\section{Discussion}
%
Afshani, Peyman, and Chan~\cite{afshani2017instance} give an
instance-optimal classical algorithm for finding 2D convex hulls,
which adapts the algorithm
by Kirkpatrick and Seidel~\cite{kirkpatrick} to run in 
$O(n(\mathcal{H}(S)+1))$ time, where $\mathcal{H}(S)$ is 
a classical-computing structural entropy measure 
that is never more than $O(\log h)$.
An interesting direction for future work would be to determine if there
is a quantum convex hull algorithm that is instance optimal 
based on a quantum structural entropy measure.

\clearpage 

\bibliographystyle{plainurl}
\bibliography{refs}
\begin{appendix}
\section{A Review of a Classical Maxima Set Algorithm}
For completeness, we describe a classical maxima set algorithm in
Algorithm~\ref{alg:classical-maxima-set},
assuming the points have distinct $x$-coordinates.

\begin{algorithm}[ht]
\caption{ClassicalMaximaSet($S$, $M$)}
\label{alg:classical-maxima-set}
\begin{algorithmic}[1]
\State \textbf{if} {$n\le 1$} \textbf{return} $M=S$.
\State \textbf{Divide step:} Divide $S$ into $S_1$ and $S_2$ of size 
at most $\lceil n/2\rceil$ each, such that the points of $S_1$
have smaller $x$-coordinates than those in $S_2$.
\State \textbf{Conquer step:}
\State Recursively call ClassicalMaximaSet($S_1$, $M_1$).
\State Recursively call ClassicalMaximaSet($S_2$, $M_2$).
\State \textbf{Combine step:}
\State Let $p_{\rm max}$ be the point in $S_2$ with largest $y$-coordinate
\State Remove all points from $M_1$ with $y$-coordinates smaller than 
$p_{\rm max}$ 
and concatenate the list of remaining points with $M_2$, returning this
as $M$.
\end{algorithmic}
\end{algorithm}

The last step (Step 9) works because all of the points in $S_2$ have a larger $x$-coordinate than all the points in $S_1$. Therefore, if a point $p$ in $S_2$ has a larger $y$-coordinate than a point $q$ in $S_1$, then $p$ dominates $q$. 
As mentioned above, the running time 
of this classical algorithm is easily seen to be 
$O(n\log n)$.

\section{A Review of a Classical Convex Hull Algorithm}
Algorithm~\ref{alg:classical-ch} finds the upper hull of a 
set, $S$, of lexicographically sorted points in the plane, where we assume,
for the sake of simplicity, that points have distinct $x$-coordinates.

\begin{algorithm}[ht]
\caption{ClassicalUpperHull($S$, $U$)}
\label{alg:classical-ch}
\begin{algorithmic}[1]
\If {$n\le 1$} 
\State \textbf{return} $U=S$.
\EndIf
\State \textbf{Divide step:} Divide $S$ into $S_1$ and $S_2$ of size 
at most $\lceil n/2\rceil$ each, such that the points of $S_1$
have smaller $x$-coordinates than those in $S_2$.
\State \textbf{Conquer step:}
\State Recursively call ClassicalUpperHull($S_1$, $U_1$).
\State Recursively call ClassicalUpperHull($S_2$, $U_2$).
\State \textbf{Combine step:}
\State 
\label{step-lp}
Find a \emph{bridge} upper tangent edge, $e=(v,w)$,
such that $v\in S_1$ and $w\in S_2$ 
no point of $S$ is above the line $\overline{vw}$.
\State Remove all points from $U_1$ (resp., $U_2$) below $e$
and concatenate the list of remaining points of $U_1$ with $e$ and 
the remaining points of $U_2$, returning this
as $U$.
\end{algorithmic}
\end{algorithm}

By a well-known point-line duality,
the bridge edge in Step~\ref{step-lp}
can be found by a solving a two-dimensional linear program;
see, e.g., \cite{deberg,seidel,orourke}.
In this duality,
each point $p$ with coordinate $(p_x, p_y)$ in the primal plane maps to a line $y = p_x x - p_y$ in the dual plane and vice versa. 
See Figure~\ref{fig:duality}.

\begin{figure}[hbt!]
    \centering
    \includegraphics[width=.9\linewidth]{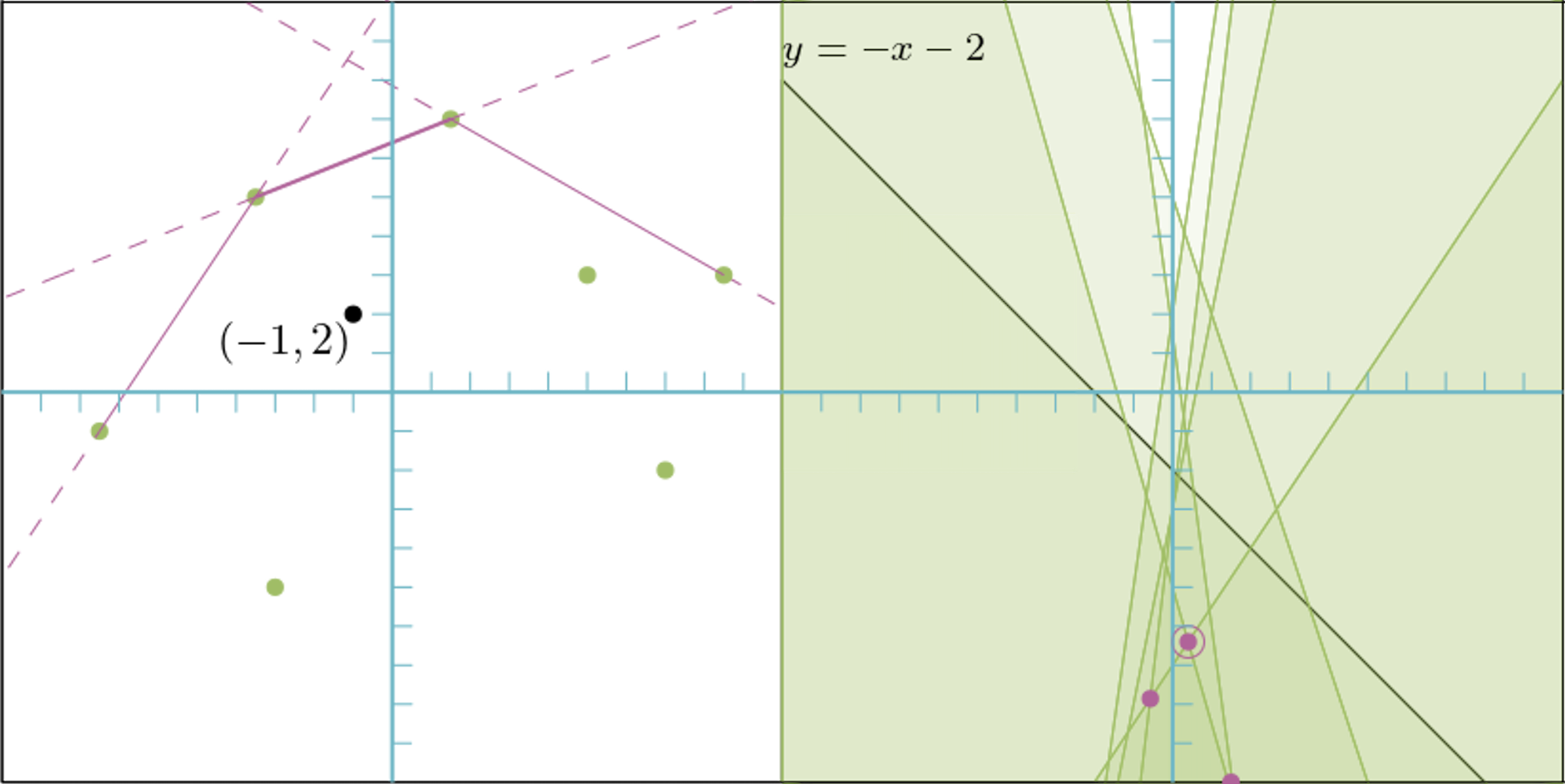}
    \caption{An example of point-line duality for the convex hull problem.
The left configuration is the primal plane containing the points of interest in green, and the lines containing the convex hull in purple. 
The solid portions are edges of the convex hull, 
and we would like to find the edge that intersects the blue vertical line $L$. 
As mentioned, each point $p$ with coordinate $(p_x, p_y)$ in the primal plane maps to a line $y = p_x x - p_y$ in the dual plane and vice versa. 
An example of this mapping is shown in black in the right configuration. 
The purple lines containing the upper hull of the points in the primal plane map to the lower envelope of the lines in the dual plane indicated by purple dots.}
    \label{fig:duality}
\end{figure}

\section{Quantum Jarvis March Convex Hull Algorithms}
As mentioned above,
another well-known classical algorithm for computing a two-dimensional
convex hull is the \emph{Jarvis march} algorithm;
see, e.g., \cite{deberg,seidel,preparata}.
In this algorithm, we start with a point, $p$, known to be on the convex 
hull, such as a point with smallest $y$-coordinate. Then, we find
the next point on the convex hull in a clockwise order, by a simple
linear-time maximum-finding step. We iterate this search, winding around
the convex hull, until we return to the starting point.
Since each iteration takes $O(n)$ time, this algorithm 
runs in $O(nh)$ time.

Lanzagorta and Uhlmann \cite{lanzagorta2004quantum}
describe a quantum assisted version of the Jarvis march algorithm.
As in the classical version,
this quantum Jarvis march convex hull algorithm starts with an edge, $(p, q)$,
known to be on the convex hull, and then does searches over the remaining 
points to find the point $r$ such that the angle formed between $(p, q)$ 
and $(q, r)$ is maximized. 
In the quantum implementation, each iteration
is performed using a Grover maximum-finding search~\cite{grover1996fast}, which runs in
$\tilde{O}(\sqrt{n})$ time.
Thus, in this quantum implementation, we
perform $h$ Grover searches for the point maximizing the angle 
until we find the starting point. 
For completion, we summarize this method in
Algorithm~\ref{alg:q-jarvis-march1} and characterize
its performance in Lemma~\ref{lemma:q-jarvis-march1}.

\begin{algorithm}
\caption{QuantumJarvisMarch($U_{in}$)}\label{alg:q-jarvis-march1}
\begin{algorithmic}[1]
\Require $U_{in}$: a unitary that encodes the data $[d_0, \ldots, d_{n-1}]$
\State Prepare a superposition of the indices and data using $U_{in}$.
\State Find an edge of the convex hull $(p_s, p_c)$.
\State Let $p_f = p_s$ denote our starting (and final) convex hull point.
\State {\bf while} $p_c \neq p_f$ {\bf do}
\State \xxx Grover search for the point $p_r$ that maximizes the angle between $(p_s, p_c)$ and $(p_c, p_r)$. 
\State \xxx Add $p_r$ to the set of output convex hull points for $S$.
\State \xxx Let $(p_s, p_c) = (p_c, p_r)$. 
\end{algorithmic}
\end{algorithm}

\begin{lemma}[Lanzagorta and Uhlmann \cite{lanzagorta2004quantum}]\label{lemma:q-jarvis-march1}
Let $S$ be a set of $m$ points.
The convex hull of $S$ in the clockwise direction 
can be computed in $\Tilde{O}(h\sqrt{m})$ time using a quantum computer, 
where $h$ is the number of points on the convex hull.
\end{lemma}

In our quantum convex hull algorithm, we use a method similar to 
this quantum Jarvis march algorithm as a subroutine, albeit
in a non-standard way.

\section{A Lower bound for Quantum Convex Hulls}

In this appendix, we show that our algorithm is optimal up to polylogarithmic factors assuming a sorted input and oracle access to the data. A key ingredient of this proof is the lower-bound on unstructured search given oracle access to a database, summarized below. 

\begin{theorem}\label{thm:search-lower-bound}
    Let $S$ be a set of size $N$ and $f: S \rightarrow \{0, 1\}$ be a Boolean function such that $f(x) = 1$ for a unique $x \in S$ and $U_{in}$ be a unitary providing oracle access to $f$. Then, to find $x$, a quantum computer requires at least $\Omega\left(\sqrt{N}\right)$ queries to $U_{in}$. 
\end{theorem}

For the lower bound, we consider a lexicographically sorted set of $n$ points whose convex hull contains $h + 1$ points lying on a parabola, and $h$ points lying above that parabola. 
We also assume a setting where the points on the convex hull appear every $n/h$ points, which we will call partition points. Furthermore, there will be a point on the convex hull between every pair of partition points. What remains is to find these remaining points on the convex hull. This problem reduces to finding the nearest point to the parabola below the chord formed by consecutive partition points. By Theorem 9, if each block has a unique winner to be found, we require at least $\Omega\left(\sqrt{n/h}\right)$ queries to find a winner per block. Since there are $h$ blocks, we see that we require $\Omega\left(\sqrt{nh}\right)$ queries to the oracle to find the convex hull. 
(See Figure \ref{fig:lower-bound}).

\begin{figure}[htb]
    \centering
    \includegraphics[width=.8\linewidth]{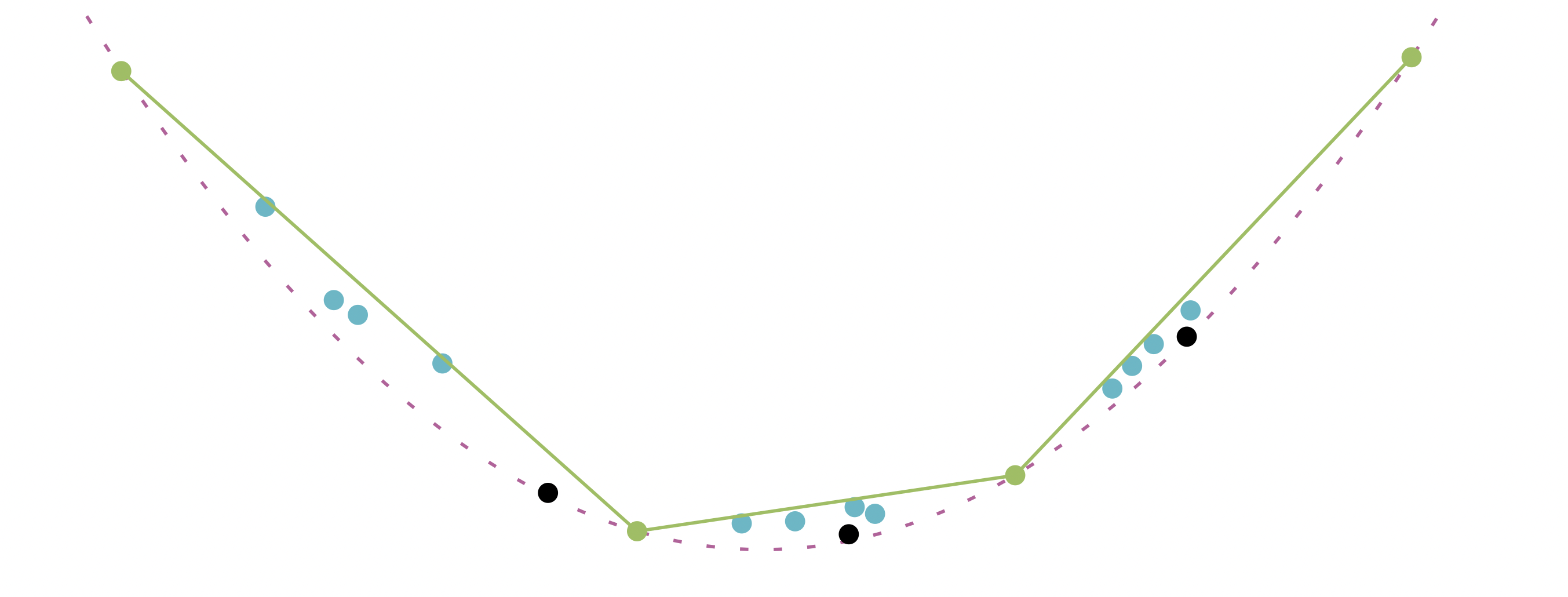}
    \caption{A set of $n = 19$ points with 7 points on the convex hull. There are $h + 1= 4$ points that lie on the parabola (in green) assumed to be given to us, and $h$ points on the convex hull lying between the chords and parabola. Now the objective is to find within each block, the point marked in black. }
    \label{fig:lower-bound}
\end{figure}

\end{appendix}

\end{document}